\documentclass[12pt]{elsarticle}%
\usepackage{amsmath}
\usepackage{amsthm}
\usepackage{amssymb}
\usepackage{amsfonts}
\usepackage{graphicx}
\usepackage{epstopdf}
\usepackage{color}
\usepackage{hyperref}
\usepackage{fullpage}
\usepackage{multirow}
\usepackage{booktabs}
\usepackage{float}%
\providecommand{\U}[1]{\protect\rule{.1in}{.1in}}
\newtheorem{theorem}{Theorem}
\newtheorem{lemma}[theorem]{Lemma}

\newtheorem{definition}[theorem]{Definition}
\newtheorem{notation}[theorem]{Notation}

\newtheorem{remark}[theorem]{Remark}

\allowdisplaybreaks
\begin{document}

\begin{frontmatter}
\title{ Majorization and  Inequalities among \\Complete Homogeneous Symmetric Functions}
\author[jx]{Jia Xu\corref{cor1}\footnote{Accepted for publication in European Journal of Combinatorics. DOI: 10.1016/j.ejc.2026.104389}}
\ead{j.jia.xu@gmail.com}
\author[yy]{Yong Yao}
\ead{yongyao525@163.com}
\cortext[cor1]{Corresponding author.}
\address[jx]{Department of Mathematics, Southwest Minzu University, Chengdu, Sichuan 610225, China}
\address[yy]{Chengdu Institute of Computer Applications, Chinese
Academy of Sciences, Chengdu, Sichuan 610213, China}
\begin{abstract}
\noindent
Inequalities among symmetric functions are fundamental in various branches of mathematics, thus motivating a systematic study of their structure. Majorization has been shown to characterize inequalities among  commonly used symmetric  functions, except for \emph{complete homogeneous} symmetric functions (shortened as CHs).
In~2011, Cuttler, Greene, and Skandera posed a natural question: Can majorization also characterize  inequalities among  CHs? Their work demonstrated that majorization characterizes  inequalities among CHs up to degree 7 and suggested exploring its validity for higher degrees.
In this paper, we show that, for every degree greater than 7, majorization does {\em not} characterize inequalities among CHs.
\end{abstract}
\end{frontmatter}

\section{Introduction}

Inequalities among symmetric functions  arise naturally in various
branches of mathematics with applications in science and engineering.  A central
 challenge is to characterize when such inequalities hold.
This challenge has been  addressed through  a variety of  techniques spanning diverse fields, such as algebra
\cite{HLP1952}, analysis \cite{Bu1988, Ti2003, Ti2021}, and combinatorics \cite{Ma1998, St1999}.

In 1902, Muirhead, in his celebrated work \cite{Mu1902}, established that
\emph{majorization} (see Definition \ref{de:Ma}, also known as the dominance order)  provides a systematic way to characterize inequalities
among \emph{monomial} symmetric functions  (see Definition~\ref{def:ssp}). We will illustrate this result through a few  simple examples.
 Consider the following monomial symmetric functions of degree $3$.%
\begin{align*}
m_{3,(3,0,0)}  & =x_{1}^{3}+x_{2}^{3}+x_{3}^{3},\\
m_{3,(2,1,0)}  & =x_{1}^{2}x_{2}+x_{1}^{2}x_{3}+x_{2}^{2}x_{1}+x_{2}^{2}%
x_{3}+x_{3}^{2}x_{1}+x_{3}^{2}x_{2},\\
m_{3,(1,1,1)}  & =x_{1}x_{2}x_{3}.%
\end{align*}
The  term-normalization (see Definition \ref{def:tn}) of the functions above is as follows.
\begin{align*}
M_{3,(3,0,0)}   =\frac{1}{3}m_{3,(3,0,0)},\ \
M_{3,(2,1,0)}   =\frac{1}{6}m_{3,(2,1,0)},\ \
M_{3,(1,1,1)}   =\frac{1}{1}m_{3,(1,1,1)},%
\end{align*}
where the set $\left\{(3,0,0),\ (2,1,0),\ (1,1,1)\right\}$ consists of all degree-3 partitions (see Definition~\ref{def:pt}),
which are weakly decreasing sequences of nonnegative integers.
 Consider  the following potential inequalities among them:
\begin{align*}
A  & :M_{3,(1,1,1)}\geq M_{3,(2,1,0)}.\\
B  & :\ M_{3,(3,0,0)}\geq M_{3,(2,1,0)}.%
\end{align*}
We would like to know whether each potential inequality  actually holds, that is,  it is true for all non-negative values of the variables
$x_1,x_2$ and $x_3$. It is easy to see that  $A$ does not hold, as shown by the following counterexample.
\begin{equation*}
    A(1,1,0): \quad \frac{0}{1} \geq \frac{2}{6} \quad \text{(which is of course false)}.
\end{equation*}
In contrast, it is not easy to check whether  $B$ holds.  This could be checked using complex and general decision algorithms, such as QEPCAD \cite{CH1991}, BOTTEMA \cite{BY2016}, TSDS \cite{XY2012}, or RealCertify \cite{ME2018}. However, these methods can be  highly time-consuming, especially when the number of variables and the degree increase.
Here Muirhead made a breakthrough by showing that majorization  provides an easy check.
The definition of majorization~(Definition~\ref{de:Ma}) yields
\[
(3,0,0)\succeq(2,1,0)\succeq(1,1,1).
\]
Muirhead's theorem then implies that

\begin{itemize}
\item $A$ does not hold because $(1,1,1)\not\succeq(2,1,0)$.

\item $B$ holds because $(3,0,0)$ $\succeq(2,1,0)$.
\end{itemize}

\noindent In general, Muirhead \cite{Mu1902} proved
that majorization completely characterizes inequalities among monomial symmetric functions:
\begin{theorem}[Muirhead]
For every $d\geq 1$, we have
\[
\forall{\mu,\lambda\in Par\left(d\right)   ,}\; \;\;\; \left(\;\forall n\;(M_{n,\mu}\geq M_{n,\lambda})\ \ \ \ \iff\ \ \ \ \ \mu
\succeq\lambda\;\right),
\]
where $Par\left(d\right)$ denotes the set of all partitions of $d$ (see Definition~\ref{def:pt}).
\end{theorem}

 In 2011,
Cuttler, Greene, and Skandera \cite{CGS2011} initiated an investigation into whether
majorization can also completely characterize inequalities among  other
commonly used  symmetric functions, such as elementary, power-sum, Schur, and complete homogeneous symmetric
functions (see Definition~\ref{def:ssp}). They proved that it  does so  among \emph{elementary}
symmetric functions and \emph{power-sum} symmetric functions. In 2016,  Sra
\cite{Sr2016} proved  that it also does so among \emph{Schur} functions. In 2021,
Khare and Tao \cite{KT2021} wrote down the first weak majorization inequalities using
Schur polynomials and further strengthened the Cuttler–Greene–Skandera implication
relating Schur polynomials and majorization. Subsequently,
McSwiggen and Novak \cite{MN2022} extended majorization inequalities to other Lie types.
More recently, Chen, Khare, and Sahi \cite{CKS2025} (conjecturally) characterized majorization
and weak majorization inequalities for Jack polynomials, a result that would subsume the elementary,
monomial, and Schur cases.

As we will only use two families of symmetric functions below, we direct the reader to \cite{Sm2024}, \cite{St1999} for basic
definitions and results on other families of symmetric functions, and to \cite{CGS2011} for classical results
on majorization inequalities in the literature.

However, it was not known whether majorization can also completely characterize inequalities
among \emph{complete homogeneous} symmetric functions,
the  remaining commonly  used symmetric polynomial
functions. Cuttler, Greene, and Skandera \cite{CGS2011} proved that majorization implies these
inequalities:

\begin{theorem}[Cuttler--Greene--Skandera]\label{thm:CGS}
For every $d\geq 1,$~ we have
\[
\forall{\mu,\lambda\in Par\left(d\right)   ,}\; \;\;\;\left(\; \forall n\;(H_{n,\mu}\geq H_{n,\lambda}) \ \ \ \ \Longleftarrow\ \ \ \ \ \mu \succeq\lambda\; \right).
\]
\end{theorem}
\noindent It remained to determine whether the converse holds,
namely whether the inequalities imply majorization, that is,
\begin{equation*}
\forall{\mu,\lambda\in Par\left(d\right)  ,}\; \;\;\;\left(\  \forall
n\ (H_{n,\mu}\geq H_{n,\lambda})\;\;\;\implies\;\;\;\mu\succeq\lambda
\ \right).
\end{equation*}
In what follows, we denote the above statement by $C(d)$. In the same paper,
Cuttler et al. proved the following result:
\begin{theorem}[Cuttler--Greene--Skandera]
For every $d\leq 7,$~ the statement $C(d)$ is true.
\end{theorem}
\noindent They left open the question on the truth of the  statement $C(d)$ for $d\geq 8$.
In this paper, we prove the following theorem:

\medskip

\noindent {\bf Theorem}(Main Result).
For every $d\geq 8,$~ the statement $C(d)$ is false.\footnote{
The works in \cite{HS2021,HS2019} proved $H_{3,\left(
4,4\right)  }\geq H_{3,\left(  5,2,1\right)  })\ \wedge\;\left(  4,4\right)
\not \succeq \left(  5,2,1\right)  $. Based on this, the authors claimed that
it is a counterexample for $C(8)$. However, it is not a
counterexample to $C(8)$. Instead, it is a counterexample to a
related, yet distinct statement shown below, which we will denote as $C^{\prime}\left(  8\right)  $:%
\[
\forall{\mu,\lambda\in Par\left(  8\right)  ,\ \ }\left(  \exists
n\ (H_{n,\mu}\geq H_{n,\lambda})\;\;\;\implies\;\;\;\mu\succeq\lambda\right)
.
\]
In contrast, we show for every $d\geq 8$ and two carefully chosen $\mu_{d} \not \succeq \lambda_{d}$
in $Par(d)$, that $H_{n,\mu_{d}}\geq H_{n,\lambda_{d}}$ on $[0,\infty)^{n}$ for all $n$---as opposed to merely
$n=3$ in \cite{HS2021,HS2019}. Our proof also differs from the approach in \emph{loc.\ cit.}, which used a sum-of-squares method.
}
\medskip

\noindent Hence majorization does \emph{not}
completely characterize the inequalities among complete homogeneous symmetric functions,
unlike the other commonly used symmetric functions.

We use the following proof strategy. For each degree $d \geq8$, we judiciously choose
special~$\mu$ and $\lambda$
and show that $\mu\not \succeq \lambda$ and
$\forall n\ (H_{n,\mu} \geq
H_{n,\lambda})$.
It is straightforward  to check  $\mu\not \succeq \lambda$.
Thus the main difficulty lies  in
proving $\forall n\ (H_{n,\mu} \geq H_{n,\lambda})$. For that, we employ an
inductive approach. First, we prove the claim for $d = 8$ by induction on $n
\geq2$, reducing the problem to the polynomial optimization  problem on the
standard simplex (see~(\ref{pb:op})). Then, we extend the result to $d > 8$
using a relaxation method (Lemma~\ref{le:ml}).

  We hope that the finding in this paper inspires future research  on
the following question: What could be a \emph{relaxation} of majorization that could completely
characterize the inequalities among complete homogeneous symmetric functions?

The remainder of the paper is structured as follows. Section~\ref{sec:pre}
introduces necessary definitions and notations. Section~\ref{sec:main}
presents the main result and provides a proof.

\section{Preliminaries}

\label{sec:pre}

We review several standard definitions and notations that were used in the
introduction and that will be used in the next sections. Readers who are
already familiar with them can choose to skip this section and refer back to
it later if needed.

\begin{definition}
[Partition, Chapter 1.1 of \cite{Ma1998}]\label{def:pt}Let $d\in\mathbb{N^{+}}$. The set of all
partitions of $d$, denoted by $Par(d)$, is defined by%
\[
Par(d)=\left\{  \left(  \lambda_{1},\ldots,\lambda_{d}\right)  \in\mathbb{N}%
^{d}:\lambda_{1}\geq\cdots\geq\lambda_{d}\geq0\ \,\text{and }\lambda_{1}%
+\cdots+\lambda_{d}=d\right\}  .
\]
\end{definition}

\begin{remark}
\

\begin{enumerate}
\item We will delete $0$ included in the elements of a partition if there is
no confusion. For example, $(2,1,0)$ can be written briefly as $(2,1)$.

\item If $m>1$ consecutive entries of the partition $\lambda$ are equal to an integer $k$,
we write this block as $k^{m}$ for brevity.
For example, $(2,1,1,1)$ can be written as $(2,1^{3})$.
\end{enumerate}
\end{remark}

\begin{definition}
[Majorization, \cite{MOA2011}, p.~8]\label{de:Ma} Let $\mu,\lambda\in Par(d)$.
We say that $\mu$ majorizes $\lambda,\ \,$ and write $\mu\succeq\lambda$,\ if%
\[
\forall{1\leq j\leq d-1} \; \left(  \sum_{i=1}^{j}\mu_{i}\;\;\geq
\;\;\sum_{i=1}^{j}\lambda_{i}\right)  .
\]

\end{definition}

\begin{definition}
[Commonly used Symmetric functions,  \cite{ Sm2024, St1999}]\label{def:ssp} \
Fix an integer $n\ge1$, and let $\lambda=(\lambda_1,\dots,\lambda_k)$ be a partition.

\medskip

(i) \emph{Monomial symmetric functions} $m_{n,\lambda}$:
\[
m_{n,\lambda}=\sum_{\pi}x_{\pi(1)}^{\lambda_{1}}x_{\pi(2)}^{\lambda_{2}}\cdots
x_{\pi(k)}^{\lambda_{k}}\ \ \ (k\leq n),
\]
where the sum ranges over all distinct permutations $\pi$ of $\{1,\dots,n\}$
such that all terms
$x_{\pi(1)}^{\lambda_1}\cdots x_{\pi(k)}^{\lambda_k}$ are distinct.
\medskip

(ii) \emph{Power-sum symmetric functions} $p_{n,\lambda}$:
\[
p_{n,\lambda}=\prod_{j=1}^{k}p_{n,\lambda_{j}},\quad  \text{where }\ p_{n,\lambda
_{j}}=\sum_{1\leq i\leq n}x_{i}^{\lambda_{j}}. %
\]

\medskip

(iii) \emph{Complete homogeneous symmetric functions} $h_{n,\lambda}$:
\[
h_{n,\lambda}=\prod_{j=1}^{k}h_{n,\lambda_{j}},\quad  \text{where }\ h_{n,\lambda
_{j}}=\sum_{1\leq i_{1}\leq\cdots\leq i_{\lambda_{j}}\leq n}x_{i_{1}}\cdots
x_{i_{\lambda_{j}}}\ \ \ \ (\text{with}\ \ h_{n,0}=1).
\]
\end{definition}

\noindent We note that, unlike classical treatments of symmetric functions in infinitely many variables,
here we consider only finitely many variables, which allows evaluation on a finite-dimensional positive orthant, following Newton \cite{Ne1732},
Muirhead \cite{Mu1902}, Cuttler–Greene–Skandera \cite{CGS2011}, and other works on majorization inequalities in the literature.

\begin{definition}[Term-normalization, \cite{CGS2011}]\label{def:tn}
Let $f(x)$ be any symmetric function defined above.
Its \emph{term-normalization}, denoted by $F(x)$, is
\[
  F(x)=\frac{f(x)}{f(1,\dots,1)}.
\]
For a family $\{f_{\mu} : \mu\in Par(d)\}$ of symmetric functions, we write $F_{\mu}$ for the corresponding term-normalized functions.
\end{definition}

\medskip

\noindent\textbf{Notation.}

We denote the standard simplex and its interior by
\[
  \Delta_n = \{ x \in [0,\infty)^n : x_1 + \cdots + x_n = 1 \},
  \qquad
  \Delta_n^\circ = \Delta_n \cap (0,\infty)^n .
\]

For \(n,d\ge1\) and \(\mu,\lambda\in Par(d)\), we write
\[
  F_{n,\mu} \ge F_{n,\lambda}
  \quad\text{if}\quad
  F_{n,\mu}(x)\ge F_{n,\lambda}(x)\quad\text{for all }x\in[0,\infty)^n .
\]

\section{Main result}

\label{sec:main} In this section, we state and prove the main result of the paper.

\begin{notation}
Let $C\left(  d\right)  $ denote the following statement:
\[
\forall{\mu,\lambda\in Par(d)  ,\ \ }\left(  \ \forall
n\ (H_{n,\mu}\geq H_{n,\lambda})\;\;\;\implies\;\;\;\mu\succeq\lambda
\ \right)  .
\]

\end{notation}

\begin{theorem}
[Main Result]\label{Th:mt}For every $d\geq8$, the statement $C\left(  d\right)
$ is false.
\end{theorem}

\noindent Before presenting the technical details, we first provide an
overview of the proof structure to help the reader grasp the overall strategy.
\bigskip

\noindent\textsc{Top-level Structure of the Proof:}
For each fixed $d\geq8$, it suffices to find $\mu,\lambda\in Par(d)$ such that
$\mu\not \succeq \lambda$ but $(H_{n,\mu}\geq H_{n,\lambda})$ on $[0,\infty)^n$
for each $n$. We now propose such $\mu,\lambda$:
\[
\mu := (2^{\lfloor d/2 \rfloor},\, 1^{\, d - 2\lfloor d/2 \rfloor}),
\qquad
\lambda:=(3,1^{d-3}),
\qquad
d\ge 8.
\]
\noindent It is clear that $\mu\not \succeq \lambda$, since $\mu_1<\lambda_1$.
However, showing that $(H_{n,\mu}\geq H_{n,\lambda})$ for
each $n$ is non-trivial, and is the
content of this work. Here, we provide a bird's-eye view.
\begin{itemize}
\item We first prove the claim for the degree $d=8$ by induction on the number
of variables $n\geq2$, transforming the problem into a polynomial
optimization problem on the standard simplex (see~(\ref{pb:op})). The details are provided in
Lemma~\ref{le:d8} and its proof.

\medskip

\item Briefly put, the proof of Lemma~\ref{le:d8} is divided into two cases,
depending on where a minimizer lies: the interior and the boundary of the
simplex. When it is in the interior, we reduce the number of variables from
$n$ to~$2,$ by exploiting the symmetry of the equations arising from Lagrange
multiplier theorem (see Lemma~\ref{le:lgr} and Lemma~\ref{le:dh-h}). When it
is on the boundary, we reduce the number of variables from $n$ to $n-1.$

\medskip

\item We then extend the result to the degrees $d>8$ by repeatedly using a
relaxation method. The details are given in Lemma \ref{le:ml} and its proof.
\end{itemize}

\noindent This concludes the top-level structure of the proof. We now present the details.

\bigskip

\noindent The formula in the following Lemma~\ref{le:dh-h} will be used in the
proof of the subsequent Lemma~\ref{le:lgr}.

\begin{lemma}\label{le:dh-h}
For $k\in\mathbb{N}$ and $i\in \{1,\ldots,n\}$,
\[
\frac{\partial h_{n,k}}{\partial x_i} = \sum_{j=0}^{k-1} h_{n,j}\,x_i^{\,k-1-j}.
\]
\end{lemma}

\begin{proof}
This identity is standard; it follows directly from differentiating the generating function of
$h_{n,k}$ with respect to $x_i$ and comparing coefficients
(see \cite[Chapter I]{Ma1998} and \cite[Chapter 7]{St1999} for the generating function method).
\end{proof}

\begin{lemma}
\label{le:lgr}Let $J_{n}=H_{n,\left(  2^{4}\right)  }-H_{n,\left(
3,1^{5}\right)  }$. If $p$ is a minimizer of $J_{n}$ over $\Delta
_{n}^{\circ}$, then $p\ $has at most two distinct components, that is,
\[
p \text{ is a permutation of } (t^u,r^v)=(\underset{u}{\underbrace{t,\dots,t}},\underset{v}{\underbrace{r,\dots,r}}),
\]
for some $t,r\in\mathbb{R}_{+}$ and some $u,v\in\mathbb{N}$ such that $u+v=n$.
\end{lemma}

\begin{proof}
We divide the proof into several steps to enhance clarity.

\begin{enumerate}

\item $p\in\Delta_{n}^{\circ}$ implies $h_{n,1}(p)-1=0$ by Definition \ref{def:ssp}. Thus,
by applying the Lagrange multiplier theorem (\cite[p.~383]{Ap1974}), we have
\begin{equation}
\exists{\lambda\in\mathbb{R}},\; \forall{i},\; \left(  W_{i}\left(  p\right)
=0\right)  , \label{eq:langrange}%
\end{equation}
where%
\[
W_{i}=\frac{\partial J_{n}}{\partial x_{i}}-\lambda\frac{\partial h_{n,1}%
}{\partial x_{i}}.
\]

\item In this step, we will repeatedly rewrite and simplify $W_{i}$. Firstly, by
recalling the definition $J_{n}=H_{n,\left(  2^{4}\right)
}-H_{n,\left(  3,1^{5}\right)  }.$ We have
\[
W_{i}=\frac{\partial H_{n,\left(  2^{4}\right)  }}{\partial x_{i}}%
-\frac{\partial H_{n,\left(  3,1^{5}\right)  }}{\partial x_{i}%
}-\lambda\frac{\partial h_{n,1}}{\partial x_{i}}.\
\]
\newline Then, by recalling
\[
H_{n,\lambda}=\frac{1}{\binom{n+\lambda_{1}-1}{\lambda_{1}}\cdots
\binom{n+\lambda_{d}-1}{\lambda_{d}}}\ h_{n,\lambda},
\]
we have%
\[
W_{i}=\frac{\partial\frac{h_{n,2}^{4}}{\binom{n+1}{2}^{4}}}{\partial x_{i}}%
-\frac{\partial\frac{h_{n,3}h_{n,1}^{5}}{\binom{n+2}{3}\binom{n}{1}^{5}}%
}{\partial x_{i}}-\lambda\frac{\partial h_{n,1}}{\partial x_{i}}.
\]
Now, we differentiate the expression above using the chain and product rules,
obtaining%
\begin{equation}
\label{eq:W}W_{i}=\frac{4h_{n,2}^{3}\frac{\partial h_{n,2}}{\partial x_{i}}%
}{\binom{n+1}{2}^{4}}-\frac{\left(  \frac{\partial h_{n,3}}{\partial x_{i}%
}\right)  \left(  h_{n,1}^{5}\right)  +\left(  h_{n,3}\right)  \left(
5h_{n,1}^{4}\frac{\partial h_{n,1}}{\partial x_{i}}\right)  }{\binom{n+2}%
{3}^{1}\binom{n}{1}^{5}}-\lambda\frac{\partial h_{n,1}}{\partial x_{i}}.
\end{equation}
By applying Lemma \ref{le:dh-h} to $k=1,2,3,$ we have%
\begin{equation}
\label{eq:dh}%
\begin{array}
[c]{l}%
\frac{\partial h_{n,1}}{\partial x_{i}} =h_{n,0},\\
\frac{\partial h_{n,2}}{\partial x_{i}} =h_{n,0}x_{i}+h_{n,1},\\
\frac{\partial h_{n,3}}{\partial x_{i}} =h_{n,0}x_{i}^{2}+h_{n,1}x_{i}%
+h_{n,2}.
\end{array}
\end{equation}
By plugging $\left(  \ref{eq:dh}\right)  $ into $\left(  \ref{eq:W}\right)  $
we have%
\[
W_{i}=\frac{4h_{n,2}^{3}\left(  h_{n,0}x_{i}+h_{n,1}\right)  }{\binom{n+1}{2}^{4}%
}-\frac{\left(  h_{n,0}x_{i}^{2}+h_{n,1}x_{i}+h_{n,2}\right)  h_{n,1}%
^{5}+5h_{n,3}h_{n,1}^{4}h_{n,0}}{\binom{n+2}{3}\binom{n}{1}^{5}}-\lambda
h_{n,0}.
\]
Finally, by collecting in the \textquotedblleft explicit\textquotedblright%
\ powers of $x_{i},$ we have%
\begin{equation}
W_{i} = ax_{i}^{2}+bx_{i}+c, \label{eq:SE}%
\end{equation}
where%
\begin{align*}
a  &  =-\frac{h_{n,0}h_{n,1}^{5}}{\binom{n+2}{3}\binom{n}{1}^{5}},\\
b  &  =\frac{4h_{n,2}^{3}h_{n,0}}{\binom{n+1}{2}^{4}}-\frac{h_{n,1}h_{n,1}%
^{5}}{\binom{n+2}{3}\binom{n}{1}^{5}},\\
c  &  =\frac{4h_{n,2}^{3}h_{n,1}}{\binom{n+1}{2}^{4}}-\frac{h_{n,2}h_{n,1}%
^{5}+5h_{n,3}h_{n,1}^{4}h_{n,0}}{\binom{n+2}{3}\binom{n}{1}^{5}}-\lambda
h_{n,0}.
\end{align*}
Note that $a$ and $b$ depend on $x$, while $c$ depends on both $x$ and
$\lambda$.

\item By combining $\left(  \ref{eq:langrange}\right)  $ and $\left(
\ref{eq:SE}\right)  ,$ we have
\begin{equation}
{\exists}{\lambda\in\mathbb{R}},\: {\forall}{i}, \left(  a\left(  p\right)
p_{i}^{2}+b\left(  p\right)  p_{i}+c\left(  p,\lambda\right)  =0\right)  .
\label{p2}%
\end{equation}

\item Note that $a\left(  p\right)  ,b\left(  p\right)  ,c\left(
p,\lambda\right)  $ do $\emph{not}$ depend on the index $i$. This motivates
the introduction of the following object:%
\[
w_{p,\lambda}\left(  y\right)  =a\left(  p\right)  y^{2}+b\left(  p\right)
y+c\left(  p,\lambda\right)  \in\mathbb{R}\left[  y\right]  ,
\]
where $y$ is a new indeterminate. Then we can rewrite (\ref{p2}) as%
\[
\exists{\lambda\in\mathbb{R}},\; \forall{i},\; \left(  w_{p,\lambda}\left(
p_{i}\right)  =0\right)  .
\]

\item It says that every $p_{i}$ is a root of $w_{p,\lambda}$. Note that
$\deg_{y}w_{p,\lambda}=2$ since $a\left(  p\right)  \neq0$ (immediate from
$h_{n,0}\left(  p\right)  =1$ and $h_{n,1}\left(  p\right)  =p_{1}%
+\cdots+p_{n}=1$). Hence $w_{p,\lambda}$ has at most two solutions. Therefore
finally we conclude that $p\ $has at most two distinct components.\qedhere
\end{enumerate}
\end{proof}

\begin{lemma}
[Degree $d=8$]\label{le:d8} We have
\[
\forall n\; \left(H_{n,\left(  2^{4}\right)  }\geq H_{n,\left(
3,1^{5}\right)  }\right).
\]

\end{lemma}

\begin{proof}
Let $J_{n}=H_{n,\left(  2^{4}\right)  }-H_{n,\left(  3, 1^{5}\right)  }$. We will prove $J_{n}\geq0$ by induction on $n$. For $n=1$, we have
\[
J_1 = H_{1,(2^4)} - H_{1,(3,1^5)} = h_{1,(2^4)} - h_{1,(3,1^5)} = x^{8} - x^{8} = 0.
\]
Therefore, we assume $n\geq2$ from this point onward.

\medskip

\noindent\textsf{Induction base}: We will show that $J_{2}\geq0$. By factoring
$J_{2}$, using a computer algebra system~\footnote{
\url{https://github.com/XuYao7/Computation.git}}, we obtain
\[
J_{2}(x_{1},x_{2})=(x_{1}-x_{2})^{2}P(x_{1},x_{2}),
\]
where
\[
P\left(  x_{1},x_{2}\right)  =\frac{1}{10368}\left(  47(x_{1}^{6}+x_{2}%
^{6})+120(x_{1}^{5}x_{2}+x_{1}x_{2}^{5})+177(x_{1}^{4}x_{2}^{2}+x_{1}^{2}%
x_{2}^{4})+176x_{1}^{3}x_{2}^{3}\right)  .
\]
Thus $J_{2}\geq0$ holds.

\medskip

\noindent\textsf{Induction step}: Assume that $J_{n-1}\geq0$ for $n\geq3$. It
suffices to show that $J_{n}\geq0$. Since $J_{n}$ is homogeneous, it suffices
to show that
\begin{equation}\label{pb:op}
\min_{x\in\Delta_{n}}J_{n}(x)\geq0.
\end{equation}

\noindent Since $\Delta_{n}$ is compact, there exists an element $p\in
\Delta_{n} $ such that $J_{n}(p)=\min\limits_{x\in\Delta_{n}} J_{n}(x)$. It
suffices to show $J_{n}(p)\geq0$. We consider the following two cases.

\begin{enumerate}
\item $p\in\Delta_{n}^{\circ}$ (the interior of $\Delta_{n}$).

By  Lemma \ref{le:lgr}, we know that $p$ is in the form $(t^{u}%
,r^{v})$ for some $t,r$ and $u+v=n$. Moreover,
$t,r>0$ since $p\in\Delta_{n}^{\circ}$. Hence it suffices to show%

\[
\forall u, v \geq0\left(  u + v = n\right)  , \; \forall t, r \in
\mathbb{R}_{+},\; \left(  J_{n}(t^{u}, r^{v}) \geq0\right).
\]

Since $J_{n}$ is homogeneous, it suffices to show%
\[
\forall u, v \geq0\left(  u + v = n\right)  ,\; \forall t\in\mathbb{R}_{+},\;
\left(  J_{n}(t^{u}, 1^{v})\geq0\right)  .
\]

Using the recursive formula
\[
h_{n,k} = \frac{1}{k} \sum_{i=1}^{k} h_{n,k-i} p_{n,i} \quad (\text{see \cite[Formula 6.2; 9]{Li1950}}),
\]
where \( p_{n,i} = \sum_{1\leq j\leq n} x_j^i \),
we rewrite \( J_n \) in terms of \( p_{n,i} \) for \( i=1,2,3 \) as follows:
\begin{equation} \label{eq:Jn-pn}
\begin{aligned}
J_n &= \frac{h_{n,2}^4}{\binom{n+1}{2}^4} - \frac{h_{n,3} h_{n,1}^5}{\binom{n+2}{3} \binom{n}{1}^5} \\
&= \frac{\left( \frac{1}{2} (p_{n,1}^2 + p_{n,2}) \right)^4}{\binom{n+1}{2}^4}
- \frac{ \frac{1}{3} \left( \frac{1}{2} (p_{n,1}^2 + p_{n,2}) p_{n,1} + p_{n,1}p_{n,2} + p_{n,3} \right) p_{n,1}^5 }{ \binom{n+2}{3} \binom{n}{1}^5 },
\quad \text{note } h_{n,0}=1\\
&= \frac{\left( \frac{1}{2} (p_{n,1}^2 + p_{n,2}) \right)^4}{\binom{u+v+1}{2}^4}
- \frac{ \frac{1}{6} p_{n,1}^8 + \frac{1}{2} p_{n,1}^{6} p_{n,2} + \frac{1}{3} p_{n,1}^{5}p_{n,3} }{ \binom{u+v+2}{3} \binom{u+v}{1}^5 },
\quad \text{since } n = u+v.
\end{aligned}
\end{equation}

Evaluating \( p_{n,i} \)  at \( (t^{u}, 1^v) \), we obtain
\begin{equation} \label{eq:p1-3}
p_{n,i}(t^{u}, 1^v) = ut^i + v,\  (i=1,2,3).
\end{equation}

It follows from (\ref{eq:Jn-pn}) and (\ref{eq:p1-3}) that $J_{n}(t^{u},1^{v})$ can be expressed as a rational function in $t,u,v$.

\bigskip

Observe that $J_{n}$ can be factored as $\hat{J}_{n} \breve{J}_{n}$, where
\[
\hat{J}_{n}=\frac{uv(t-1)^{2}}{(u+v+2)(u+v+1)^{4}(u+v)^{6}}\ ,
\]
and $\breve{J}_{n}$ is a polynomial in $t,u,v$. It is clear that $\hat{J}_{n}$ is
non-negative. Thus it is sufficient to show
\begin{equation}
\label{J1}\forall u, v \geq0 \left(  u + v = n\right)  ,\; \forall
t\in\mathbb{R}_{+},\; \breve{J}_{n}(t^{u},1^{v})\geq0.
\end{equation}

It can be challenging to check whether condition (\ref{J1}) holds by directly inspecting the coefficients of
 $\breve{J}_{n}$ in $t$, due to the
presence of negative terms. To address this difficulty, we introduce the following approach.

\medskip

Note that for $u=0$ or $v=0$, we have $J_{n}=\hat{J}_{n}\breve{J}_{n}=0$. Thus,
it suffices to show
\[
\forall u, v \geq1\left(  u + v = n\right)  ,\; \forall t\in\mathbb{R}_{+},\;
\breve{J}_{n}(t^{u},1^{v})\geq0.
\]
Setting $u=k+1$ and $v=\ell+1$, it then suffices to show%
\begin{equation}
\label{J2}\forall{k,\ell\geq0} \left(  k+\ell+2=n \right)  ,\; \forall
t\in\mathbb{R}_{+},\; \left(  \breve{J}_{n}(t^{k+1},1^{\ell+1})\geq0\right)  .
\end{equation}

By using a computer algebra system~\footnote{\url{https://github.com/XuYao7/Computation.git}}, we found
the following expression for $\breve{J}_{n}(t^{k+1},1^{\ell+1})$:
\[
\breve{J}_{n}(t^{k+1},1^{\ell+1}) = \sum_{i=0}^{6}
c_{i}t^{i},
\]
where the coefficients $c_{i}$ are polynomials in $k$ and $l$ with positive coefficients.
For instance,
{\footnotesize
\[
c_{6}= \left(  k +2\right)  \left(  k +1\right)  ^{3} \left(  k^{4}+2
k^{3} l +k^{2} l^{2}+12 k^{3}+17 k^{2} l +5 k \,l^{2}+49 k^{2}+43 k l +5
l^{2}+82 k +32 l +47\right).
\]
}
For completeness, the full expressions of the coefficients $c_i$ are given in~\ref{app:ci}.

\smallskip

\noindent Note that all the coefficients $c_{i}$ are positive. Hence the condition
(\ref{J2}) holds. Therefore, we conclude that $J_{n}(p)\geq0$.

\bigskip

\item $p\in\partial\Delta_{n}$ (the boundary of $\Delta_{n}$).

Since $J_{n}$ is symmetric and $p\in\partial\Delta_{n}$, we can assume,
without losing generality, that
\[
p=(\tilde{p},0).
\]

By the induction assumption, we have $J_{n-1}(\tilde{p})\geq0.$ Thus, it
suffices to show that%
\[
k_{1}J_{n}(p)\geq k_{2}J_{n-1}(\tilde{p})\geq0 \;\; \text{for some }
k_{1},k_{2}>0.
\]
We will choose $\displaystyle k_{1}=\binom{n+2}{3}\binom{n}{1}^{5}\,\text{and
}k_{2}=\binom{n+1}{3}\binom{n-1}{1}^{5}$. Note that%
\begin{align*}
k_{1}J_{n}(p)  & =\binom{n+2}{3}\binom{n}{1}^{5}\,\left(  \frac{h_{n,\left(
2^{4}\right)  }\left(  p\right)  }{\binom{n+1}{2}^{4}}-\frac{h_{n,\left(
3,1^{5}\right)  }\left(  p\right)  }{\binom{n+2}{3}\binom{n}{1}^{5}}\right)
\\
& =\underset{T\left(  n\right)  }{\underbrace{\frac{\binom{n+2}{3}\binom{n}%
{1}^{5}}{\binom{n+1}{2}^{4}}}}h_{n,\left(  2^{4}\right)  }\left(  p\right)
-h_{n,\left(  3,1^{5}\right)  }\left(  p\right)\\
& =T\left(  n\right)  )  h_{n,\left(  2^{4}\right)  }\left(
p\right)  -h_{n,\left(  3,1^{5}\right)  }\left(  p\right).
\end{align*}
Likewise note that%
\begin{align*}
k_{2}J_{n-1}(\tilde{p})  & =\binom{n+1}{3}\binom{n-1}{1}^{5}\left(
\frac{h_{n-1,\left(  2^{4}\right)  }\left(  \tilde{p}\right)  }{\binom{n}%
{2}^{4}}-\frac{h_{n-1,\left(  3,1^{5}\right)  }\left(  \tilde{p}\right)
}{\binom{n+1}{3}\binom{n-1}{1}^{5}}\right)  \\
& =\underset{T\left(  n-1\right)  }{\underbrace{\frac{\binom{n+1}{3}%
\binom{n-1}{1}^{5}}{\binom{n}{2}^{4}}}}h_{n-1,\left(  2^{4}\right)  }\left(
\tilde{p}\right)  -h_{n-1,\left(  3,1^{5}\right)  }\left(  \tilde{p}\right) \\
& =T\left(  n-1\right)  )  h_{n,\left(  2^{4}\right)  }\left(
p\right)  -h_{n,\left(  3,1^{5}\right)  }\left(  p\right),\quad
\text{since } h_{n,\lambda}\left(p\right)  =h_{n-1,\lambda}\left(\tilde{p}\right).
\end{align*}
Combining these expressions, we obtain
\begin{align*}
  & k_{1}J_{n}(p)-k_{2}J_{n-1}(\tilde{p})\\
= & \Big(T\left(  n\right)  h_{n,\left(
2^{4}\right)  }\left(  p\right)  -h_{n,\left(  3,1^{5}\right)  }\left(
p\right)  \Big)-\Big(T\left(  n-1\right)  h_{n,\left(  2^{4}\right)  }\left(
p\right)  -h_{n,\left(  3,1^{5}\right)  }\left(  p\right)  \Big)\\
= & \Big(T\left(  n\right)  -T\left(  n-1\right)  \Big)h_{n,\left(
2^{4}\right)  }\left(  p\right).
\end{align*}
Hence, it suffices to show that $T\left(  n\right)  \geq T\left(  n-1\right)
$. For this, let us simplify $T\left(  n\right)  $:%
\[
T\left(  n\right)  =\frac{\binom{n+2}{3}\binom{n}{1}^{5}}{\binom{n+1}%
{2}^{4}}=\frac{\frac{\left(  n+2\right)  \left(  n+1\right)  \left(  n\right)
}{3\cdot2\cdot1}n^{5}}{\left(  \frac{\left(  n+1\right)  \left(  n\right)
}{2\cdot1}\right)  ^{4}}=\frac{8}{3}\frac{n^{3}+2n^{2}}{\left(  n+1\right)
^{3}}.
\]
Viewing $n\geq3$ as a real number, it suffices to show $T^{\prime}\left(
n\right)  \geq0$. Note
\[
T^{\prime}\left(  n\right)  =\frac{8}{3}\frac{\left(  3n^{2}+4n\right)
\left(  n+1\right)  ^{3}-\left(  n^{3}+2n^{2}\right)  3\left(  n+1\right)
^{2}}{\left(  n+1\right)  ^{6}}=\frac{8}{3}\frac{n\left(  n+4\right)
}{\left(  n+1\right)  ^{4}}\geq0.
\]
We conclude that $J_{n}\left(  p\right)  \geq0$.
\end{enumerate}
\end{proof}

\begin{remark}
The proof of Lemma~\ref{le:d8} does not work for  general partitions. The reason is that,
in the case $p \in \Delta_{n}^{\circ}$, the formula
\[
\forall u, v \geq0\left(  u + v = n\right)  , \; \forall t, r \in
\mathbb{R}_{+},\; \left(  J_{n}(t^{u}, r^{v}) \geq0\right)
\]
does not hold for general partitions. For example, for partitions $(2^{3})$ and $(3, 1^{3})$,
 $J_{n}=H_{n,\left(  2^{3}\right)  }-H_{n,\left(  3, 1^{3}\right)  }\geq0$
 fails for $n \geq 3$ when $t=2$, $r=1$, $u=1$, and $v=n-1$.
\end{remark}

\begin{lemma}
\label{le:ml} Let $m\geq4$. We have

\begin{enumerate}
\item $\forall n \left(  H_{n,(2^{m})}\geq H_{n,(3,1^{2m-3})}\right)  $, and

\item $\forall n \left(  H_{n,(2^{m},1)}\geq H_{n,(3,1^{2m-2})}\right)  $.
\end{enumerate}
\end{lemma}

\begin{proof}
We prove each claim.

\begin{enumerate}
\item
Let
$F_{n,m}=\frac{H_{n,(2^{m})}}{H_{n,(3,1^{2m-3})}}.$
From Lemma \ref{le:d8}, we have $F_{n,4}\geq1$.
We also have%
\begin{equation}
F_{n,m}\geq F_{n,m-1}. \label{eq:02}%
\end{equation}
Since
\begin{align*}
\frac{F_{n,m}}{F_{n,m-1}}  &  =\frac{\frac{H_{n,(2^{m})}}{H_{n,(3,1^{2m-3})}}}%
{\frac{H_{n,(2^{m-1})}}{H_{n,(3,1^{2m-5})}}}%
=\frac{\frac{\binom{n+2}{3}\binom{n}{1}^{2m-3}}{\binom{n+1}{2}^{m}}\frac{(h_{n,2})^{m}}{h_{n,3}\ (h_{n,1})^{2m-3}}}%
{\frac{\binom{n+2}{3}\binom{n}{1}^{2m-5}}{\binom{n+1}{2}^{m-1}}\frac{(h_{n,2})^{m-1}}{h_{n,3}\ (h_{n,1})^{2m-5}}}%
=\frac{h_{n,2}}{\binom{n+1}{2}}\frac{\binom{n}{1}^{2}}{(h_{n,1})^{2}}\\
&  =\frac{H_{n,(2)}}{H_{n,(1^{2})}} \geq 1 \text{\;\;\;by Theorem \ref{thm:CGS}} .
\end{align*}

Iterating inequality (\ref{eq:02}) and noting that $F_{n,4}\geq1$ by Lemma \ref{le:d8}, we have
\[
F_{n,m}\geq F_{n,{m-1}}\geq\cdots\geq F_{n,4}\geq1.
\]
Hence
\[
\forall n \left(  H_{n,(2^{m})}\geq H_{n,(3,1^{2m-3})}\right)  ,\ \text{for}\ m\geq4.
\]

\smallskip

\item
Note%
\[
\frac{H_{n,(2^{m},1)}}{H_{n,(3,1^{2m-2})}}%
=\frac{H_{n,(2^{m})}\ H_{n,\left(  1\right)  }}{H_{n,(3,1^{2m-3})}\ H_{n,\left(  1\right)  }}=\frac{H_{n,(2^{m})}%
}{H_{n,(3,1^{2m-3})}}=F_{n,m}\geq1.
\]

Hence
\[
\forall n \left(  H_{n,(2^{m},1)}\geq H_{n,(3,1^{2m-2})}\right)  ,\ \text{for}\ m\geq4.
\]

\end{enumerate}
\end{proof}

\noindent Finally we are ready to prove Theorem~\ref{Th:mt}.

\begin{proof}
[Proof of Theorem \ref{Th:mt}] Let $d\geq8$. We consider two cases depending
on the parity of $d$.

\begin{enumerate}
\item Consider $d=2m$. Take $\mu=(2^{m}),\lambda=(3,1^{2m-3})$. \ It is easy to see that
\[
\mu=(2^{m})=(\underset{m}{\underbrace{2,\ldots,2}})\ \ \nsucceq
\ \ (3,\underset{2m-3}{\underbrace{1,\ldots,1}})=(3,1^{2m-3}%
)=\lambda.
\]
However, from Lemma \ref{le:ml}, we have
\[
H_{n,\mu}=H_{n,(2^{m})}\ \ \geq\ \ H_{n,(3,1^{2m-3}%
)}=H_{n,\lambda}, \text{ for every}\ n.
\]

\item Consider $d=2m+1$. Take $\mu=(2^{m},1),\lambda=(3,1^{2m-2})$. It is easy to see that
\[
\mu=(2^{m},1)=(\underset{m}{\underbrace{2,\ldots,2},1})\ \ \nsucceq
\ \ (3,\underset{2m-2}{\underbrace{1,\ldots,1}})=(3,1^{2m-2}%
)=\lambda.
\]
However, from Lemma \ref{le:ml}, we have
\[
H_{n,\mu}=H_{n,(2^{m},1)}\ \ \geq\ \ H_{n,(3,1^{2m-2}%
)}=H_{n,\lambda}.
\]

\end{enumerate}

\noindent Hence we have proved that
\[
\forall{d\geq8,\ \ }\exists{\mu,\lambda\in Par\left(  d\right)  ,\ \ \ }%
\forall n\ (H_{n,\mu}\geq H_{n,\lambda})\;\wedge\;\mu\not \succeq \lambda.
\]

Thus, the preceding arguments show that the condition $C(d)$ fails for every $d\ge 8$.
This completes the proof of Theorem~\ref{Th:mt}.
\end{proof}

\appendix
\section{Explicit expressions of the coefficients $c_i$}
\label{app:ci}
{\footnotesize
\begin{align*}
c_{6}  &  =\left(  k +2\right)  \left(  k +1\right)  ^{3} \left(  k^{4}+2
k^{3} l +k^{2} l^{2}+12 k^{3}+17 k^{2} l +5 k l^{2}+49 k^{2}+43 k l +5
l^{2}+82 k +32 l +47\right)  ,\\
c_{5}  &  =2 \left(  k +2\right)  \left(  k +1\right)  ^{3} \left(  3 k^{3} l
+6 k^{2} l^{2}+3 k l^{3}+2 k^{3}+32 k^{2} l +37 k l^{2}+7 l^{3}+21
k^{2}+106 k l +52 l^{2}+64 k \right. \\
& \quad \left. +109 l +60\right)  ,\\
c_{4}  &  =\left(  l +1\right)  \left(  k +1\right)  ^{2} \left(  15 k^{4} l
+30 k^{3} l^{2}+15 k^{2} l^{3}+11 k^{4}+173 k^{3} l +208 k^{2} l^{2}+46 k
l^{3}+121 k^{3}+677 k^{2} l\right. \\
& \quad \left. +426 k l^{2}+35 l^{3}+442 k^{2}+1074 k l +272 l^{2}+662 k +599 l +354\right)  ,\\
c_{3}  &  =4 \left(  l +1\right)  ^{2} \left(  k +1\right)  ^{2} \left(  5
k^{3} l +10 k^{2} l^{2}+5 k l^{3}+6 k^{3}+53 k^{2} l +53 k l^{2}+6
l^{3}+51 k^{2}+157 k l +51 l^{2} \right. \\
& \quad \left. +125 k+125 l +88\right)  ,\\
c_{2}  &  =\left(  l +1\right)  ^{2} \left(  k +1\right)  \left(  15 k^{3}
l^{2}+30 k^{2} l^{3}+15 k l^{4}+46 k^{3} l +208 k^{2} l^{2}+173 k l^{3}+11
l^{4}+35 k^{3}+426 k^{2} l \right. \\
& \quad \left.  +677 k l^{2}+121 l^{3}+272 k^{2}+1074 k l +442 l^{2}+599 k +662 l +354\right)  ,\\
c_{1}  &  =2 \left(  l +2\right)  \left(  l +1\right)  ^{3} \left(  3 k^{3} l
+6 k^{2} l^{2}+3 k l^{3}+7 k^{3}+37 k^{2} l +32 k l^{2}+2 l^{3}+52
k^{2}+106 k l +21 l^{2} \right. \\
& \quad \left. +109 k +64 l +60\right)  ,\\
c_{0}  &  =\left(  l +2\right)  \left(  l +1\right)  ^{3} \left(  k^{2}
l^{2}+2 k l^{3}+l^{4}+5 k^{2} l +17 k l^{2}+12 l^{3}+5 k^{2}+43 k l +49
l^{2}+32 k +82 l +47\right)  .
\end{align*}
}
\bigskip

\noindent\textbf{Acknowledgements}. The authors are grateful to Bi-can Xia for
drawing their attention to some relevant references and to Hoon Hong for
helpful conversations. Special thanks to the anonymous reviewers for their
time, effort, and valuable input in improving the quality of this manuscript.
This work was supported by the Fundamental Research Funds for the Central Universities, Southwest
Minzu University (No. ZYN2025111).


\end{document}